\pdfoutput=1
\documentclass[12pt,a4paper]{amsart}
\usepackage[a4paper,inner=2.5cm,outer=2.5cm,top=2.5cm,bottom=2.5cm]{geometry}
\usepackage{amsmath,amssymb,amsthm,enumerate,mathtools,stmaryrd
}
\usepackage{hyperref}
\hypersetup{colorlinks=true,linkcolor=blue,citecolor=teal,filecolor=magenta,urlcolor=cyan}

\usepackage[T1]{fontenc}
\usepackage[utf8]{inputenc}

\usepackage{makecell}
\usepackage{graphicx}

\usepackage{float}

\usepackage{extarrows}

\usepackage{xcolor}

\newcommand{\bb}[1]{{\textcolor{blue}{#1}}}

\newcommand{\C}{\mathbb{C}}

\usepackage{url}

\usepackage[backend=bibtex,style=alphabetic,sorting=nyt,isbn=false,url=false,doi=true,maxalphanames=10,minalphanames=4,mincitenames=4,maxcitenames=10,minnames=4,maxnames=10,giveninits=true,maxbibnames=99]{biblatex}
\theoremstyle{plain}
\newtheorem{theorem}{Theorem}[section]
\newtheorem{proposition}[theorem]{Proposition}
\newtheorem{corollary}[theorem]{Corollary}

\theoremstyle{definition}
\newtheorem{definition}[theorem]{Definition}

\newtheorem{notation}[theorem]{Notation}

\makeatletter
\@addtoreset{proofpart}{theorem}
\makeatother
\theoremstyle{remark}
\newtheorem{remark}[theorem]{Remark}

\newcommand{\Z}{\mathbb{Z}}

\newcommand{\cD}{\mathcal{D}}
\newcommand{\cS}{\mathcal{S}}

\newcommand{\cW}{\mathcal{W}}
\newcommand{\cT}{\mathcal{T}}

\newcommand{\VEVc}[1]{{\big\langle \! 0 \big| {#1} \big| 0 \!\big\rangle^\circ}}

\newcommand{\res}{\mathop{\rm res}}

\newcommand{\restr}[2]{\mathop{\big\lfloor_{{#1}\to {#2}}}}
\newcommand{\set}[1]{\llbracket {#1} \rrbracket}

\def\Expr{\mathsf{Expr}}

\numberwithin{equation}{section}

\title
{Topological recursion, symplectic duality, and generalized fully simple maps}

\author[A.~Alexandrov]{A.~Alexandrov}
\address{A.~A.: Center for Geometry and Physics, Institute for Basic Science (IBS), Pohang 37673, Korea
}
\email{alex@ibs.re.kr}

\author[B.~Bychkov]{B.~Bychkov}
\address{B.~B.: Department of Mathematics, University of Haifa, Mount Carmel, 3498838, Haifa, Israel}
\email{bbychkov@hse.ru}

\author[P.~Dunin-Barkowski]{P.~Dunin-Barkowski}
\address{P.~D.-B.: Faculty of Mathematics, HSE University, Usacheva 6, 119048 Moscow, Russia; HSE--Skoltech International Laboratory of Representation Theory and Mathematical Physics, Skoltech, Bolshoy Boulevard 30 bld. 1, 121205 Moscow, Russia; and NRC “Kurchatov Institute” -- ITEP, 117218 Moscow, Russia}
\email{ptdunin@hse.ru}

\author[M.~Kazarian]{M.~Kazarian}
\address{M.~K.: Faculty of Mathematics, HSE University, Usacheva 6, 119048 Moscow, Russia; and Igor Krichever Center for Advanced Studies, Skoltech, Bolshoy Boulevard 30 bld. 1, 121205 Moscow, Russia}
\email{kazarian@mccme.ru}

\author[S.~Shadrin]{S.~Shadrin}
\address{S.~S.: Korteweg-de Vries Institute for Mathematics, University of Amsterdam, Postbus 94248, 1090GE Amsterdam, The Netherlands}
\email{S.Shadrin@uva.nl}	

\begin{document}
	
\begin{abstract} For a given spectral curve, we construct a family of \emph{symplectic dual} spectral curves for which we prove an explicit formula expressing the $n$-point functions produced by the topological recursion on these curves via the $n$-point functions on the original curve. As a corollary, we prove topological recursion for the generalized fully simple maps generating functions. 
\end{abstract}
	
\maketitle

\setcounter{tocdepth}{2}
\tableofcontents

\section{Introduction}

In this paper we study relations between $n$-point functions obtained by spectral curve topological recursion~\cite{EynardOrantin-toporec} for various spectral curves.
Specifically, the notion of the so-called \emph{symplectic duality} of spectral curves was introduced in~\cite[Section~1.3.5]{BDKS-symplectic} (where by spectral curve $\mathsf{SC}$  we mean an algebraic curve $\Sigma$ equipped with two meromorphic functions $X$ and $y$, and a bidifferential $B$, with certain conditions). In the present paper we introduce a large class of pairs $(\mathsf{SC}_1,\, \mathsf{SC}_2)$ of mutually symplectic dual spectral curves for which we provide an explicit closed formula expressing a given $n$-point function obtained by applying topological recursion on $\mathsf{SC}_2$ in terms of the topological recursion $m$-point functions of $\mathsf{SC}_1$, and vice versa. In this large class of pairs of spectral curves one can take $\mathsf{SC}_1$ as any spectral curve (under mild conditions), while $\mathsf{SC}_2$ depends on $\mathsf{SC}_1$ and a function $\psi$ satisfying certain conditions. Our explicit formula relating the respective $n$-point functions is, in a way, a generalization of the respective formulas from~\cite{BDKS-FullySimple} and~\cite{ABDKS-XYSwap}.

Our results can be used in another direction as well: if one has some set of $n$-point functions for which it is known that they are related by our explicit formula to the $m$-point functions of a certain spectral curve topological recursion, then our theorem implies that they are produced by topological recursion for a particular spectral curve. As an application of this, we prove that the \emph{generalized fully simple maps generating functions} satisfy topological recursion.

Fully simple maps are natural objects in algebraic topology, see~\cite{borot2019relating,BGF-SimpleMaps} for the precise definition, and for discussion on why they are interesting. In~\cite{BGF-SimpleMaps} it was conjectured that the generating functions for the weighted numbers of the fully simple maps satisfy 
 topological recursion, 
  which was later proved in~\cite{borot2021topological} and in~\cite{BDKS-FullySimple}. In the latter paper a natural generalization of these generating functions was proposed (which has a combinatorial meaning too, see below), and it was conjectured that these generalized generating functions also satisfy topological recursion, see~\cite[Conjecture 4.4]{BDKS-FullySimple}. This conjecture got a new impulse as a natural source of examples of symplectic duality mentioned above, see~\cite[Section 1.3.5]{BDKS-symplectic}. In the present paper we prove an extended version of this conjecture.



\subsection{Two formulations of topological recursion}\label{sec:toprecedefs}

We present two formulations of topological recursion. The first one (the original formulation) is needed only for the definition of spectral curve symplectic duality in Section \ref{sec:SCSD}.
The second one is its reformulation in terms of the loop equations and the projection property~\cite{BS-blobbed}.

\subsubsection{Original definition}\label{sec:trorig}


The Chekhov--Eynard--Orantin topological recursion~\cite{ChekhovEynard,ChekhovEynard2,EynardOrantin-toporec,EynardOrantin-toporec-enum,EynardSurvey} associates to an input that consists of
\begin{itemize}
	\item a compact Riemann surface $\Sigma$ with a chosen basis of $\mathfrak{A}$- and $\mathfrak{B}$-cycles (the so-called spectral curve);
	\item two meromorphic functions $X$ and $y$ on $\Sigma$ such that all critical points $p_1,\dots,p_N\in \Sigma$ of $X$ are simple, $y$ is regular at $p_i$ and $dy|_{p_i}\not=0$, $i=1,\dots,N$;
	\item a symmetric bi-differential $B$ on $S^2$ with a double pole on the diagonal with bi-residue $1$ and holomorphic at all other points, normalized in such a way that its $\mathfrak A$-periods in both variables vanish (the so-called Bergman kernel)
\end{itemize}
an output that consists of a system of symmetric meromorphic differentials $\omega^{(g)}_{m,0}$ on $\Sigma^m$, $g\geq 0$, $m\geq 1$, which are called the \emph{correlation differentials}. The recursion is given by setting $\omega^{(0)}_{1,0}(z_1)  \coloneqq 
y(z_1)dX(z_1)/X(z_1)$, $\omega^{(0)}_{2,0}(z_1,z_2) \coloneqq B(z_1,z_2)$, and then for $2g-2+m>0$ we have
\begin{align} \label{eq:CEO-TR-original}
	\omega^{(g)}_{m,0}(z_{\llbracket m \rrbracket}) & = \frac 12 \sum_{i=1}^N \res_{z=p_i} \frac{\int_z^{\sigma_i (z)} \omega^{(0)}_{2,0}(z_1,\bullet) }{\omega^{(0)}_{1,0}({\sigma_i (z)}) - \omega^{(0)}_{1,0}(z)} \Bigg( \omega^{(g-1)}_{m+1,0} (z,\sigma_i(z),z_{\llbracket m \rrbracket \setminus \{1\}})
	\\ \notag & \qquad + \sum_{\substack{g_1+g_2=g \\ I_1\sqcup I_2 = \llbracket m \rrbracket \setminus \{1\}\\ (g_i,|I_i|)\not= (0,0)}} \omega^{(g_1)}_{|I_1|+1,0}(z,z_{I_1})\omega^{(g_2)}_{|I_2|+1,0}(\sigma_i(z),z_{I_2}) \Bigg).
\end{align}
Here and below we use notation $\llbracket m \rrbracket \coloneqq \{1,\dots,m\}$, $z_I = \{z_i\}_{i\in I}$ for any $I\subseteq \llbracket m \rrbracket$, and $\sigma_i$ denotes the deck transformation with respect to the function $X$ near the point $p_i$. The extra seemingly superfluous subscript $0$ in the notation $\omega^{(g)}_{m,0}$ is explained below.

\begin{remark}\label{rem:moregentoprec}
	The requirements for $X$ and $y$ can be substantially relaxed, for instance, it is sufficient to assume that $ydX/X$ is meromorphic (and the function $X$ itself might have essential singularities), and in fact for $2g-2+m>0$ the recursion only uses the local germs of $X$ and $y$ at the points $p_1,\dots,p_N$. 
	
		The setup for what concerns our main application, a proof of a generalization of~\cite[Conjecture 4.4]{BDKS-FullySimple}, is $\Sigma=\C P^1$, meromorphic $dX/X$ and $y$, and $B=dz_1dz_2/(z_1-z_2)$ in some global coordinate $z$, cf.~\cite[Section 4.1]{BDKS-symplectic}; see Section~\ref{sec:mapsdef} below.
\end{remark}

\subsubsection{Loop equation definition}\label{sec:toprecloopeq}

According to \cite{BS-blobbed}, an equivalent reformulation of the topological recursion 
goes as follows.
Let $\Sigma$, $X$, $y$ and $B$ be the same as in Section~\ref{sec:trorig} (in particular, recall that $p_1,\dots,p_N$ are the critical points of $X$).
%
Consider a system of symmetric meromorphic differentials $\omega^{(g)}_{m,0}(z_{\llbracket m\rrbracket})$, $g\geq 0$, $m\geq 1$, and define functions $W^{(g)}_{m,0}(z_{\llbracket m\rrbracket})$ via
\begin{align}\label{eq:function-diff}
	\omega^{(g)}_{m,0}(z_{\llbracket m\rrbracket})
	&=W^{(g)}_{m,0}\prod_{i=1}^m \dfrac{dX(z_i)}{X(z_i)}.
\end{align}
We introduce the auxiliary functions $W^{(g)}_{m,0}(z_{\llbracket m\rrbracket})$ simply due to the fact that some statements are easier to express in terms of functions rather than in terms of differentials. In the presentation of the definition below we follow~\cite[Section 5.1]{ABDKS-XYSwap}.

\begin{definition}\label{def:trloopeqdef} We say that the symmetric meromorphic differentials $\omega^{(g)}_{m,0}$, $g\geq 0$, $m\geq 1$, satisfy the topological recursion on the spectral curve $(\Sigma,X(z),y(z),B(z_1,z_2))$ if
	\begin{itemize}
		\item \emph{(Initial conditions)} We have:
		\begin{align} \label{eq:initial}
			\omega^{(0)}_{1,0}(z) & = y(x) d X(z)/X(z); & \omega^{(0)}_{2,0}(z_1,z_2) &= B(z_1,z_2).
		\end{align}
	(in particular, this implies that $W^{(0)}_{1,0}(z)=y(z)$).
		\item \emph{(Linear Loop Equations)} For any $g,m\geq0$  and $1\leq i \leq N$ 
		the function $W^{(g)}_{m+1,0}$ may have a pole at $p_i$ 
		 with respect to the first argument such that its principal part is skew symmetric with respect to the involution $\sigma_i$. In other words, we have that
		\begin{align}\label{eq:lloop}
			W^{(g)}_{m+1,0}(z,z_{\llbracket m\rrbracket})+ W^{(g)}_{m+1,0}(\sigma_i(z),z_{\llbracket m\rrbracket})		
		\end{align}
		is holomorphic in $z$ at $z\to p_i$.
		\item \emph{(Quadratic Loop Equations)} For any $g,m\geq 0$
		the meromorphic function
		\begin{align}\label{eq:qloop}
			W^{(g-1)}_{m+2,0}(z,z,z_{\llbracket m\rrbracket})+\sum_{\substack{g_1+g_1=g\\ I_1\sqcup I_2 = \llbracket m \rrbracket}}
			W^{(g_1)}_{|I_1|+1,0}(z,z_{I_1})	W^{(g_2)}_{|I_2|+1,0}(z,z_{I_2})	
		\end{align}
		may have a pole at $p_i$ in $z$ such that its principal part is skew symmetric with respect to the involution $\sigma_i$ for any $i=1,\dots,N$.
		\item \emph{(Projection Property)} For any $g,m\geq 0$ such that $2g-1+m>0$ the 
		differential $\omega^{(g)}_{m+1,0}(z,z_{\llbracket m\rrbracket})$ has no poles in $z$ other than $p_1,\dots,p_N$.
	\end{itemize}
\end{definition}

The projection property implies that $W^{(g)}_{m+1,0}(z,z_{\llbracket m\rrbracket})$ also has no poles in~$z$ other than $p_1,\dots,p_N$. But this weaker condition on $W^{(g)}_{m+1,0}$ is not sufficient for the projection property. We need a stronger requirement that $W^{(g)}_{m+1,0}$ vanishes at the poles of $dX/X$ so that $W^{(g)}_{m+1,0}(z,z_{\llbracket m\rrbracket}) \tfrac{dX(z)}{X(z)}$ is holomorphic in $z$ outside $p_1,\dots,p_N$.

The linear and quadratic loop equations together determine the principal part of the poles of $W^{(g)}_{m+1,0}(z,z_{\llbracket m\rrbracket})$ in $z$ up to the terms with at most simple poles at $p_i$. Multiplying by $dX(z)/X(z)\prod_{i=1}^m {dX(z_i)}/{X(z_i)}$ we obtain the principal part of the poles of $\omega^{(g)}_{m+1,0}(z,z_{\llbracket m\rrbracket})$ without ambiguity. Any meromorphic differential form on $\Sigma$ is determined uniquely by the principal parts of its poles and its periods. The principal relation of topological recursion is nothing but an explicit formula expressing the form $\omega^{(g)}_{m+1,0}$ in terms of the principal parts of its poles.




\subsection{Spectral curve symplectic duality}\label{sec:SCSD}
We define here a certain specialization of the general definition of symplectic duality which was given in \cite[Section 1.3.5]{BDKS-symplectic}. 
We will refer  to this specialization by \emph{$\psi$-symplectic duality}:
\begin{definition}\label{def:psidual}
	Let $(\Sigma,X,y,B)$ be a spectral curve (according to the definitions of Section~\ref{sec:toprecedefs}). Then its \emph{$\psi$-symplectic dual} curve (where $\psi$ is a function) is $(\Sigma,w,y,B)$ with
	\begin{equation}\label{eq:wofz}
		w(z)=X(z) \, e^{-\psi(y(z))}.
	\end{equation}
\end{definition}
\begin{remark}
	Note that if $X(z)$ is meromorphic then generally $w(z)$ will not be meromorphic if $\psi$ is, say, a polynomial. Then we need to be in the context of Remark~\ref{rem:moregentoprec} for $(\Sigma,w,y,B)$ to make sense as a spectral curve, since $d\log w$ is a meromorphic differential in this example. 
\end{remark}
\begin{remark}
	Note that this construction corresponds to a group action of the space of functions $\psi$ (treated as the additive abelian group) on the set of spectral curves. For a fixed $\psi$ the respective correspondence is a duality in the sense that the inverse action transforming the $\psi$-symplectic dual spectral curve back to the original one is given by the action of the function $-\psi$.
\end{remark}

Let also
\begin{equation}\label{eq:Sdef}
		\mathcal{S}(\zeta):= \dfrac{e^{\zeta/2}-e^{-\zeta/2}}{\zeta}.
\end{equation}

It turns out that under certain conditions on $\psi$ (and also mild conditions on the original curve $(\Sigma,X,y,B)$) one can express the $n$-point functions produced by topological recursion on the curve $(\Sigma,w,y,B)$ that we denote by
\begin{align}\label{eq:function-diff-w}
	\omega^{(g)}_{0,n}(z_1,\ldots,z_n)
	&=W^{(g)}_{0,n}\prod_{i=1}^n \dfrac{dw(z_i)}{w(z_i)}
\end{align}
in terms of the $m$-point functions produced by topological recursion on the curve $(\Sigma,X,y,B)$ via a certain explicit formula.

Specifically, let us state the main theorem of the present paper. Note that the statement of the theorem is followed by an extended list of remarks on its conditions. 
\begin{theorem}\label{thm:mainthm}
	Let
	\begin{align}\label{eq:hatpsi}
		\hat\psi(\theta,\hbar)&=\cS(\hbar\partial_{\theta})P(\theta) +\log R(\theta),
	\end{align}
	where	
	$R$ is an arbitrary rational function and $P$ is an arbitrary polynomial, such that $\hat\psi(0,0)=0$, and let
	\begin{equation}\label{eq:psi}
		\psi(\theta)=\hat\psi(\theta,\hbar)|_{\hbar=0}= P(\theta)+\log R(\theta).
	\end{equation}

Let $(\Sigma,X,y,B)$ be a spectral curve with the  \underline{\emph{meromorphy condition}}, that is, we assume that $d\log X$ and $y$ are meromorphic, 
and let $(\Sigma,w,y,B)$ be its $\psi$-symplectic dual curve according to Definition~\ref{def:psidual}.

We also further assume the \underline{\emph{generality condition}} that all zeroes and poles of $R(\theta)$, all zeroes of $P(\theta)$, and all zeroes of $dw$ are simple, and zeroes of $dX$ do not coincide with zeroes of $dw$.

Then there exist universal expressions 
 \begin{align}\label{eq:expr}
\Expr^{\hat\psi}_{g,n}\Big(\big\{W^{(h)}_{m,0}\big\}_{2h-2+m\leq 2g-2+n}, \big\{\tfrac{dX_i}{X_i},\tfrac{dw_i}{w_i}\big\}_{i=1,\dots,n}, \big\{\tfrac{X_iX_j}{(X_i-X_j)^2}\big\}_{\substack {i,j=1,\dots,n\\ i\not=j}},\Big),
 \end{align}
whose explicit forms are given below, such that the functions given by these expressions are meromorphic and satisfy the initial conditions and the linear and quadratic loop equations~\eqref{eq:initial},~\eqref{eq:lloop}, and~\eqref{eq:qloop} for the spectral curve data $(\Sigma,w,y,B)$. 

Introduce some additional \underline{\emph{enumerative-type spectral curve condition}} that there exists a point $O\in \Sigma$ such that $X(O)=0$ and $X$ is a local coordinate near $O$, that is, $dX(O)\not=0$, and moreover we assume that at all poles of $y$ either $X$ is finite and not ramified or $dX/X$ has a simple pole. Under this extra condition we claim that expressions~\eqref{eq:expr} multiplied by $\prod_{i=1}^n \frac{dw(z_i)}{w(z_i)}$ satisfy the projection property. In other words, 
the $n$-point functions $W^{(g)}_{0,n}$ 
produced by topological recursion on $(\Sigma,w,y,B)$ are expressed in terms of the $m$-point functions $W^{(g)}_{m,0}$ 
 produced by topological recursion on $(\Sigma,X,y,B)$ as
\begin{align}\label{eq:WExpr}
	&W^{(g)}_{0,n}
	=\Expr^{\hat\psi}_{g,n}\Big(\big\{W^{(h)}_{m,0}\big\}_{2h-2+m\leq 2g-2+n}, \big\{\tfrac{dX_i}{X_i},\tfrac{dw_i}{w_i}\big\}_{i=1,\dots,n}, \big\{\tfrac{X_iX_j}{(X_i-X_j)^2}\big\}_{\substack {i,j=1,\dots,n\\ i\not=j}},\Big).
\end{align}
\end{theorem}

Several remarks on the statement of the theorem are in order.

\begin{remark} We understand~\eqref{eq:WExpr} as a relation between $\omega^{(g)}_{0,n}$ and $\omega^{(g)}_{m,0}$, it is just more convenient to write it in terms of the auxiliary functions $W^{(g)}_{0,n}$ and $W^{(g)}_{m,0}$; it is certainly possible to write it directly in terms of $\omega$'s, but the formula for $\Expr^{\hat\psi}_{g,n}$ that we give below would be slightly more bulky.
\end{remark}

\begin{remark} \label{rem:generality}
	The \emph{generality condition} can actually be lifted. In what concerns the loop equations, the discussion will be analogous to~\cite[Section 5.2]{ABDKS-XYSwap}. And as far as the full topological recursion statement is concerned, one has to replace the topological recursion with the \emph{Bouchard--Eynard recursion} of~\cite{BouchardEynard}. This can be done in a completely analogous way to how it is discussed in \cite[Section~5.1]{BDKS-toporec-KP}, based on the taking limits procedure discussed in detail in~\cite{BBCKS}.
\end{remark}

\begin{remark} The \emph{meromorphy condition} can be relaxed for some particular choices of $\hat\psi$. A notable example is Family II in~\cite{BDKS-toporec-KP}, which is relevant for the quantum knot invariants, cf.~\cite{Homfly-2}.
\end{remark}

\begin{remark}\label{rem:zero}
	Note that many different spectral curves produce one and the same collection of $m$-point differentials $\omega^{(g)}_{m,0}$. Indeed, if one changes $X$ and $y$ in some way such that $\omega^{(0)}_{1,0}$ and the zeroes of $dX$ remain the same, and the branches of $X$ as a ramified covering swap in involution, then all $\omega^{(g)}_{m,0}$'s remain the same as well, as evident from~\eqref{eq:CEO-TR-original}. Thus, if we are interested only in the collection of $\omega^{(g)}_{m,0}$'s, then we can choose some different spectral curve which produces the same collection via topological recursion. 
	
	To this end, we remark that the \emph{enumerative-type spectral curve condition} is justified by our main application,~\cite[Conjecture 4.4]{BDKS-FullySimple}. However, the result on topological recursion might hold even if it is not the case, though we would have to sharpen the \emph{meromorphy condition} instead. 
	
	For instance, we can drop the condition that there exists this special point $O$. Then for any $O$ where $X(O)=c\not=0$ and $dX(O)\not=0$, we can use a new spectral curve data with $\tilde X = X-c$ instead of $X$, 
	with an appropriate change of $y$ such that $\omega^{(0)}_{1,0}=ydX/X$ remains unchanged. It is always possible to choose $O$ such that the function $\tilde y = y (X-c)/X$ is regular there. Then we start with this new curve $(\Sigma,\tilde X,\tilde y,B)$ which corresponds to the same $m$-point differentials $\omega^{(g)}_{m,0}$, and apply our theorem for it. Since in general we want $\tilde y$ to be meromorphic, we have to assume that $X$ is meromorphic in this modification of the setup. 
\end{remark}

The summary of the remarks above can be formulated as follows: there is a variety of statements similar to Theorem~\ref{thm:mainthm}, where the conditions can be modified in this or that way. Our choice of conditions is tailored for a particular application that we have in mind, a generalization of~\cite[Conjecture 4.4]{BDKS-symplectic} in the spirit of symplectic duality as in~\cite[Section 1.3.5]{BDKS-symplectic}, and we choose some simplifications that always hold in this case to refer as much as possible to the arguments proving the projection property in~\cite{BDKS-toporec-KP} and to avoid extra analysis of similar kind. However, the statement of the theorem is still much more general than this application. The message that we want to deliver here is that for a physically motivated application that would involve symplectic duality but possibly a slightly different set of conditions a variation of our argument would still work, just some extra analysis of emerging poles is needed. 

Theorem~\ref{thm:mainthm} is proved in Section~\ref{sec:toprec} below. The proof works as follows: we take the $n$-point functions given by~\eqref{eq:WExpr} and prove that the loop equations and the projection property (w.r.t. the curve $(\Sigma,w,y,B)$) hold (see Section~\ref{sec:toprecloopeq}). Explicitly, Equation~\eqref{eq:WExpr} is given by	
	\begin{align} \label{eq:SymplDualExplicit}
		& W_{0,n}^{(g)} (w_{\llbracket n \rrbracket}) 
		=
		[\hbar^{2g}] \sum_{\Gamma} \frac{
			\hbar^{2g(\Gamma)}}{|\mathrm{Aut}(\Gamma)|} \prod_{i=1}^n
		\sum_{k_i=0}^\infty ( w_i\partial_{w_i})^{k_i} [v_i^{k_i}]
		\\ \notag &
		\frac{d\log X_i}{d\log w_i}\;
		\restr{\theta_i}{ W^{(0)}_{1,0} (X_i)}
		\sum_{r_i=0}^\infty \Big(\partial_{\theta_i} + v_i\psi'(\theta_i)\Big)^{r_i}
		e^{v_i\left(\frac{\cS(\hbar v_i \partial_{\theta_i})}{\cS(\hbar \partial_{\theta_i})}\hat\psi(\theta_i)-\psi(\theta_i)\right)}
		[u_i^{r_i}]
		\\ \notag &
		\frac{1}{ u_i \cS(\hbar u_i)} e^{u_i \cS(\hbar u_i X_i \partial_{X_i}) \sum_{\tilde g=0}^\infty \hbar^{2\tilde g} W^{(\tilde g)}_{1,0} (X_i)-u_i W^{(0)}_{1,0} (X_i)}
		\\ \notag &
		\prod_{e\in E(\Gamma)} \prod_{j=1}^{|e|
		} \restr{(\tilde u_j, \tilde X_j) }{(u_{e(j)},X_{e(j)})} \tilde u_j \cS(\hbar \tilde u_j \tilde X_j \partial_{\tilde X_j})  \sum_{\tilde g=0}^\infty \hbar^{2\tilde g}\tilde W^{(\tilde g)}_{|e|,0}(\tilde X_{\llbracket |e|\rrbracket})
		\\ \notag &
		+\delta_{n,1}[\hbar^{2g}] \sum_{k=0}^\infty (w_1\partial_{w_1})^k [v^{k+1}]
		\restr{\theta}{ W^{(0)}_{1,0} (X_1)} e^{v\left(\frac{\cS(\hbar v \partial_{\theta})}{\cS(\hbar \partial_{\theta})}\hat\psi(\theta)-\psi(\theta)\right) } (w_1\partial_{w_1})W^{(0)}_{1,0} (X_1)
		\\ \notag &
		+ \delta_{(g,n),(0,1)} W^{(0)}_{1,0} (X_1).
	\end{align}
	Here and below we use the notation $[x^k]\sum_{\ell=-\infty}^\infty f_\ell x^\ell := f_k$; and we set 
	$\tilde W^{(0)}_{2,0}(\tilde X_1,\tilde X_2) =  W^{(0)}_{2,0}(\tilde X_1,\tilde X_2)- \frac{\tilde X_1 \tilde X_2}{(\tilde X_1-\tilde X_2)^2}$ if $e(1)=e(2)$, and $\tilde W^{(0)}_{2,0}(\tilde X_1,\tilde X_2) =  W^{(0)}_{2,0}(\tilde X_1,\tilde X_2) $ otherwise. For all $(g,n)\not=(0,2)$ we simply have $\tilde W^{(g)}_{n,0} =  W^{(g)}_{n,0}$.
	The sum is taken over all connected graphs $\Gamma$ with $n$ labeled vertices and multiedges of index $\geq 2$, where the index of a multiedge $e$ is the number of its ``legs'' and we denote it by $|e|$. For a multiedge $e$ with index $|e|$  we control its attachment to the vertices by the associated map $e\colon \llbracket |e| \rrbracket \to 
	\llbracket n \rrbracket
	$ that we denote also by $e$, abusing notation (so $e(j)$ is the label of the vertex to which the $j$-th ``leg'' of the multiedge $e$ is attached). Note that this map can be an arbitrary map from $\llbracket |e| \rrbracket$ to $\llbracket n \rrbracket$; in particular, it might not be injective, i.e.~we allow a given multiedge to connect to a given vertex several times. By $g(\Gamma)$ we denote the first Betti number of $\Gamma$. By $\restr{a}{b}$ we denote the operator of substitution $a\to b$, that is, $\restr{a}{b}f(a)=f(b)$ for any function $f$. 

	\begin{remark} For each $g\geq 0$, $n\geq 1$, the right hand side of Equation~\eqref{eq:SymplDualExplicit} is manifestly a finite sum of finite products of differential operators applied to $W^{(\tilde{g})}_{m,0}$ ($\tilde g\geq 0$, $m\geq 1$). 
	\end{remark}

	
	\begin{remark}
		In the special case $(g,n)=(0,1)$
		Equation~\eqref{eq:SymplDualExplicit} implies that $W^{(0)}_{0,1} = W^{(0)}_{1,0}$. 
		It is also straightforward to see from Equation~\eqref{eq:SymplDualExplicit} that $\omega^{(0)}_{0,2} = \omega^{(0)}_{2,0}$.	
	\end{remark}

\subsection{Generating functions of generalized fully simple maps}\label{sec:mapsdef}
Now let us introduce the so-called \emph{generalized fully simple maps generating functions}. In~\cite[Conjecture~4.4]{BDKS-FullySimple} it was conjectured that these 
generating functions 
satisfy topological recursion. We prove an extended version of this conjecture here in the present paper as a corollary of our main theorem.

For the definition of ordinary and fully simple maps as graphs on surfaces we refer to~\cite{borot2021topological}, see also~\cite{borot2019relating}. The generating functions counting these ordinary and fully simple maps can be naturally represented as the so-called Fock spaces correlators. These Fock space correlators can, in turn, be generalized in a natural way to the generalized ordinary (respectively, generalized fully simple) maps case, as we call it. In some cases these generalized ordinary and fully simple maps can also be defined in terms of more complicated graphs on surfaces, see~\cite{BCGFLS-inprep}.

\begin{remark}
	Note that the situation is as follows: we start with a counting problem for certain well-defined objects (maps and fully simple maps), then we take the respective generating functions, and then we note that they admit a natural generalization, and then we study these generalized functions. The question is whether these generalized generating functions still correspond to the counting problem for some more general objects. The paper~\cite{BCGFLS-inprep} partially answers this, describing the respective objects for some special cases of the generalization.
\end{remark}

Let us give the Fock space formulations of the problems of counting generalized ordinary and generalized fully simple maps here.

 Consider the charge zero Fock space~$\mathcal{V}_0$ whose bosonic realization is presented as $\mathcal{V}_0\cong \mathbb{C}[[q_1,q_2,q_3,\dots]]$. Define the operators $J_k$, $k\in\Z$, acting on $\mathcal{V}_0$ as $J_k=k\partial_{q_k}$, $J_{-k}=q_k$  (the operator of multiplication by $q_k$) if $k>0$, and $J_0=0$. Given a formal power series $\hat\psi_1(\theta,\hbar)$ in $\theta$ and $\hbar^2$ such that $\hat\psi_1(0,0)=0$, we introduce also the operator $\mathcal{D}_{\hat\psi_1}$ acting on $\mathcal{V}_0$ diagonally in the basis of Schur functions indexed by partitions $\lambda$, $|\lambda| \geq 0$,
\begin{equation}
	\mathcal{D}_{\hat\psi_1}\colon s_\lambda \mapsto e^{\sum_{(i,j)\in \lambda} \hat\psi_1(\hbar\,(j-i),\hbar)} s_\lambda. 	
\end{equation}
Note that $q_k$ here serve as the so-called power sum variables for the Schur functions.
Let $| 0 \rangle$ denote $1\in \mathcal{V}_0$ and let $ \langle 0 | \colon \mathcal{V}_0 \to \mathbb{C}$ be the extraction of the constant term. 

Then (cf.~\cite{BDKS-FullySimple}), for parameters $t_1,t_2,\dots$ and $s_1,s_2,\dots$,
\begin{align}\label{eq:Wmaps}
W^{\mathrm{MAPS},(g)}_{m,0}=[\hbar^{2g-2+m}]\VEVc{ \prod_{i=1}^m \Big(\sum_{j=-\infty }^\infty J_j X_i^{j} \Big) e^{\sum_{i=1}^\infty \frac{t_iJ_i}{i\hbar}}\cD_{\hat\psi_1}		e^{\sum_{i=1}^\infty \frac{s_iJ_{-i}}{i\hbar}}		}  
\end{align}
is the generating function for counting ``generalized ordinary maps'' of genus $g$ with $n$ ``distinguished polygons'',
and
\begin{align} \label{eq:Wfsmaps}
W^{\mathrm{MAPS},(g)}_{0,n}= [\hbar^{2g-2+n}]	\VEVc{ \prod_{i=1}^n \Big(\sum_{j=-\infty }^\infty J_j w_i^{j} \Big) \cD_{\hat\psi}  e^{\sum_{i=1}^\infty \frac{t_iJ_i}{i\hbar}}\cD_{\hat\psi_1}		e^{\sum_{i=1}^\infty \frac{s_iJ_{-i}}{i\hbar}}		},
\end{align}
where $\hat\psi(\theta,\hbar)$ is another formal power series in $\theta$ and $\hbar^2$ such that $\hat\psi(0,0)=0$,
is the generating function for counting ``generalized fully simple maps'' of genus $g$ with $n$ ``distinguished polygons''. Here ``$\phantom{}^\circ$'' stands for taking the ``connected correlator'' via the inclusion-exclusion formula, see e.g.~\cite{BDKS-FullySimple} for the precise definition of that.

In \cite{BDKS-FullySimple} it was conjectured that under certain conditions on $\hat\psi_1$ and $\hat\psi$ the sets of functions $W^{\mathrm{MAPS},(g)}_{m,0}$ and $W^{\mathrm{MAPS},(g)}_{0,n}$ are produced by topological recursion on two different spectral curves related by~\eqref{eq:wofz} .

\begin{remark}
	Note that \cite[Conjecture~4.4]{BDKS-FullySimple}, as formulated in that paper, deals with a quite restricted situation where $\hat\psi_1(\theta)=-\hat\psi(\theta)=\log(P(\theta))$, with $P(\theta)$ being a polynomial (and, in particular, $\hat\psi_1$ and $\hat\psi$ do not depend on $\hbar$). However, it turns out that the statement of this conjecture holds in a more general situation which is described in the next two theorems given below. 
\end{remark}

The extended version of the ordinary maps part of this conjecture was proved in~\cite{BDKS-symplectic} (see also~\cite{bonzom2022topological} for a different proof of the ordinary maps part of \cite[Conjecture~4.4]{BDKS-FullySimple}). In the present paper we prove the extended version of the fully simple part of this conjecture. Let us recall the theorem regarding the ordinary maps:
\begin{theorem}[{\cite[Theorem~5.7]{BDKS-symplectic}}] \label{thm:mapsrec}
	Let $W^{\mathrm{MAPS},(g)}_{m,0}$ be the ones of~\eqref{eq:Wmaps} with $t_i=0,\;i>d,$ and $s_i=0,\;i>e,$ for some fixed $d,e\in\mathbb{Z}_{\geq 1}$, and for 
	\begin{align}
		\hat\psi_1=\log R_1(\theta) + \cS(\hbar\partial_{\theta})P_1(\theta),
	\end{align}
	where $R_1$ is an arbitrary rational function, $P_1$ is an arbitrary polynomial, and $\cS$ is defined in~\eqref{eq:Sdef}.
	We also further assume the \emph{generality condition}, i.e. that all zeroes and poles of $R(\theta)$, all zeroes of $P(\theta)$, and all zeroes of $dX$ are simple. 
	
	Then there exist functions $X(z)$ and $\Theta(z)$ such that the $m$-point differentials
	\begin{equation}
		\omega^{\mathrm{MAPS},(g)}_{m,0}(z_1,\dots,z_m) := W^{\mathrm{MAPS},(g)}_{m,0}(X(z_1),\dots,X(z_n))\;\prod_{i=1}^m\dfrac{dX(z_i)}{X(z_i)}
	\end{equation}
	satisfy the topological recursion for the spectral curve
	\begin{align}\label{eq:mapscurve}		
		\left(\mathbb{C}P^1,X(z),\Theta(z),dz_1dz_2/(z_1-z_2)^2\right).
	\end{align}
	The functions $X(z)$ and $\Theta(z)$ can be unambiguously obtained from a certain set of implicit algebraic equations described in~\cite[Section~4.1]{BDKS-symplectic} (we omit them here for brevity); moreover,  $\Theta(z)$ is a rational function, while $d\log X(z)$ is a rational 1-form.
\end{theorem}
\begin{remark}
	We give~\cite[Theorem~5.7]{BDKS-symplectic} here in a slightly different form which is more convenient for our purposes in the present paper. The difference is in the spectral curve, as we use a slightly different definition of topological recursion in the present paper, and also we use the freedom in the choice of $X$ and $y$ as discussed in~\cite[Section~1.1.2]{BDKS-symplectic}. 
\end{remark}

\begin{remark}
	Note that in the statement of the theorem (and the one below) we rename the function $y(z)$ by $\Theta(z)$; the purpose of this is to connect the notation to the one used in~\cite{BDKS-symplectic}.
\end{remark}

\begin{remark} Remark~\ref{rem:generality} on the generality condition holds here as well.
\end{remark}

As a corollary of Theorem~\ref{thm:mainthm}, which is the main statement of the present paper, we get the following theorem on topological recursion for generalized fully simple maps:
\begin{theorem}\label{Thm:fsmapsrec}
	Let $W^{\mathrm{MAPS},(g)}_{0,n}$ be the ones of~\eqref{eq:Wfsmaps} with $t_i=0,\;i>d,$ and $s_i=0,\;i>e,$ for some fixed $d,e\in\mathbb{Z}_{\geq 1}$, and for 
	\begin{align}
		\hat\psi&=\log R(\theta) + \cS(\hbar\partial_{\theta})P(\theta),\\ 		
		\hat\psi_1&=\log R_1(\theta) + \cS(\hbar\partial_{\theta})P_1(\theta),
	\end{align}
	where $R,R_1$ are arbitrary rational functions, $P,P_1$ are arbitrary polynomials, and $\mathcal{S}(z)$ is given by~\eqref{eq:Sdef}; and let	
	\begin{equation}
		w(z) = X(z) \, e^{- \psi(\Theta(z))},
	\end{equation}
	where $X(z)$ and $\Theta(z)$ are the ones from Theorem~\ref{thm:mapsrec}, applied for the case with the same $t$'s, $s$'s and $\psi_1$.
	
	We also further assume the \emph{generality condition}, i.e. that all zeroes and poles of $R(\theta),R_1(\theta)$, all zeroes of $P(\theta),P_1(\theta)$ and all zeroes of $dw$ are simple, and zeroes of $dX$ do not coincide with zeroes of $dw$. 
	
	Then the $n$-point differentials
	\begin{equation}
		\omega^{\mathrm{MAPS},(g)}_{0,n}(z_1,\dots,z_n) := W^{\mathrm{MAPS},(g)}_{0,n}(w(z_1),\dots,w(z_n))\;\prod_{i=1}^n\dfrac{dw(z_i)}{w(z_i)}
	\end{equation}
	satisfy the topological recursion on the spectral curve
	\begin{align}\label{eq:fscurve}				
		\left(\mathbb{C}P^1,w(z),\Theta(z),dz_1dz_2/(z_1-z_2)^2\right).
	\end{align}
\end{theorem}
\begin{remark} Again, Remark~\ref{rem:generality} on the generality condition holds here as well.
\end{remark}

\begin{proof}[Proof of Theorem~\ref{Thm:fsmapsrec}]
	This theorem directly follows from Theorem~\ref{thm:mapsrec}, \cite[Theorem 4.14]{BDKS-FullySimple}, and Theorem~\ref{thm:mainthm}.
	
	Indeed, Theorem~\ref{thm:mapsrec} says that the $m$-point functions $W^{\mathrm{MAPS},(g)}_{m,0}$ satisfy topological recursion on the spectral curve~\eqref{eq:mapscurve}. Note that all conditions listed in the assumptions of Theorem~\ref{thm:mainthm} are manifestly satisfied, see~\cite[Section 4.1]{BDKS-symplectic}.
	
	Now, \cite[Theorem 4.14]{BDKS-FullySimple} says that the $n$-point functions $W^{\mathrm{MAPS},(g)}_{0,n}$ can be expressed in terms of the functions $W^{\mathrm{MAPS},(g)}_{m,0}$ exactly via formulas~\eqref{eq:SymplDualExplicit}. Since Theorem~\ref{thm:mainthm} says that the $n$-point functions produced by the topological recursion on the spectral curve~\eqref{eq:fscurve} are expressed in terms of the functions $W^{\mathrm{MAPS},(g)}_{m,0}$ via the same expressions~\eqref{eq:SymplDualExplicit}, this means that they must coincide with $W^{\mathrm{MAPS},(g)}_{0,n}$.
\end{proof}

\subsection{Organization of the paper}

In Section~\ref{sec:expansionAndMN} we express the functions $W^{(g)}_{m,0}$ as correlators on the Fock space. Then we define the $(m,n)$-point functions and the auxiliary $\cW$-functions, which are used in the proofs in Section~\ref{sec:toprec}, and state their properties.

In Section~\ref{sec:toprec} we prove the loop equations and the projection property on the spectral curve $(\Sigma,w,y,B)$, where $w(z)$ is given by~\eqref{eq:wofz}, for the right hand side of~\eqref{eq:SymplDualExplicit} assuming that the functions $W^{(g)}_{m,0}$ satisfy topological recursion on the curve $(\Sigma,X,y,B)$. Then we prove the main Theorem~\ref{thm:mainthm} (which is a simple implication of the statements above).



\subsection{Acknowledgments}
A.~A. was supported by the Institute for Basic Science (IBS-R003-D1). A.~A. is grateful to IHES for hospitality. Research of B.~B. was supported by the ISF grant 876/20. 
P.D.-B. was supported by the Russian Science Foundation (Grant No. 20-71-10073).
M.K. was supported by the International Laboratory of Cluster Geometry NRU HSE, RF Government grant, ag. № 075-15-2021-608 dated 08.06.2021. 
S.~S. was supported by the Netherlands Organization for Scientific Research. All the authors are grateful to the University of Haifa and especially to Alek Vainshtein for hospitality. We are also grateful to the anonymous referees for their remarks.



\section{
	Expansion and the 	\texorpdfstring{$(m,n)$}{(m,n)}-point functions} \label{sec:expansionAndMN}

\subsection{Expansion of the \texorpdfstring{$m$}{m}-point functions}\label{sec:expansion}
Let us work in the setup of Theorem~\ref{thm:mainthm}. Let $(\Sigma,X,y,B)$ be our spectral curve. Recall that we 
assume that there exists a point $O\in \Sigma$ such that $X$ has a simple zero at $O$ while $y$ is holomorphic at $O$; let $y_0:=y(O)$.

The function $X$ defines a local coordinate near the point $O$. By the construction of topological recursion on the spectral curve $(\Sigma,X,y,B)$, the form $\omega^{(g)}_{m,0}$ 
is holomorphic at $O$ with 
appropriate singular corrections for 
unstable cases
and we can consider the expansions of these forms in the local coordinates $X_i=X(z_i)$:
%
\begin{equation}
	\omega^{(g)}_{m,0}(z_{\llbracket m\rrbracket})
	-\delta_{(g,m),(0,2)}\tfrac{dX_1dX_2}{(X_1-X_2)^2}-\delta_{(g,m),(0,1)}
	y_0\tfrac{dX_1}{X_1}
	\sim
	\sum_{k_1,\dots,k_m=1}^\infty C^{(g)}_{k_1,\dots,k_m} \prod_{i=1}^m X_i^{k_i} \Big(\frac{dX_i}{X_i}\Big).
\end{equation}

From now on we abuse the notation a little bit, using the same symbols for the expansion of $\omega$'s and $W$'s as for the their global versions. 
The above expansion can be represented in the following equivalent way
\begin{align}\label{eq:omegacor}
	& \sum_{g=0}^{\infty} \hbar^{2g-2+m}\omega^{(g)}_{m,0} -\delta_{m,2}\tfrac{dX_1dX_2}{(X_1-X_2)^2}-\delta_{m,1}\hbar^{-1}
	y_0\tfrac{dX_1}{X_1}
	= \VEVc{\prod_{i=1}^m \Big( \sum_{\ell_i=1}^\infty X_i^{\ell_i}\Bigl(\frac{dX_i}{X_i}\Bigr) J_{\ell_i}\Big) \;Z},
	\\
	& Z\coloneqq  \exp\left(\sum_{\substack{g\geq 0,\ m\geq 1\\ k_1,\dots,k_m\geq 1}}\tfrac{\hbar^{2g-2+m}}{m!}
	C^{(g)}_{k_1,\dots,k_m} \prod_{j=1}^m \frac{J_{-k_i}}{k_i}\right).\label{def:Z}
\end{align}

In order to avoid carrying the $(0,1)$ correction term through all the formulas below, we assume from now on without any loss of generality that $y_0=0$.

%
We have
\begin{align} \label{eq:WZ}
	W^{(g)}_{m,0}(X_{\llbracket m\rrbracket}) = [\hbar^{2g-2+m}]\VEVc{ \prod_{i=1}^m \Big(\sum_{j=-\infty}^\infty J_j X_i^{j} \Big) 	Z		}.
\end{align}
Note that here, as opposed to~\eqref{eq:omegacor}, in the right hand side the sums run from $-\infty$ rather then from $1$. This difference precisely corresponds to the term $\delta_{(g,m),(0,2)}\tfrac{dX_1dX_2}{(X_1-X_2)^2}$ in the left hand side of~\eqref{eq:omegacor}, see~\cite[Proposition 4.1]{BDKS-OrlovScherbin}.

Now let us forget that we wanted to define $W^{(g)}_{0,n}$ as the $n$-point functions of the $\psi$-symplectic dual spectral curve, and instead let us use the following definition for their formal power series expansions:
\begin{align}  \label{eq:WsdZ}
	W^{(g)}_{0,n}(w_{\llbracket n\rrbracket}) := [\hbar^{2g-2+n}]	\VEVc{ \prod_{i=1}^n \Big(\sum_{j=-\infty}^\infty J_j w_i^{j} \Big) \cD_{\hat\psi} 	Z		}.
\end{align}
\cite[Theorem 4.14]{BDKS-FullySimple} implies that these $n$-point formal power series are given exactly by the expressions~\eqref{eq:SymplDualExplicit} in terms of the formal power series $W^{(g)}_{m,0}$ of~\eqref{eq:WZ}. Since the formal power series $W^{(g)}_{m,0}$'s are expansions of some globally defined functions, the formal power series $W^{(g)}_{0,n}$'s are like this as well. We extend to $W^{(g)}_{0,n}$'s our convention to abuse the notation and use the same symbols for the functions and their formal power series expansions, though for them this means expansions near $O$  in the variable $w$ defined as 
\begin{equation}
	w(z)=X(z) \, e^{-\psi(y(z))}.
\end{equation}
Note that under our assumptions on $X$, $y$ and $\psi$ (in the setup of Theorem~\ref{thm:mainthm}), we have
that $w(O)=0$ and $w$ is a local coordinate near $O$.
Then, in order to prove Theorem~\ref{thm:mainthm}, we want to prove that the $W^{(g)}_{0,n}$ of \eqref{eq:WsdZ} satisfy the topological recursion on the spectral curve $(\Sigma,w,y,B)$.


\subsection{\texorpdfstring{$(m,n)$}{(m,n)}-point functions}
Similarly to~\cite{BDKS-symplectic}, let us define the $(m,n)$-point functions $W^{(g)}_{m,n}$:

\begin{definition}
	\begin{align}		
	&W^{(g)}_{m,n}(X_{\set{m}};w_{\set{m+n}\setminus\set{m}}) \\ \nonumber
	&\phantom{W}:= [\hbar^{2g-2+m+n}]	\VEVc{ \prod_{i=m+1}^{m+n} \Big(\sum_{j=-\infty}^\infty J_j w_i^{j} \Big)\; \cD_{\hat\psi}\; \prod_{i=1}^{m} \Big(\sum_{j=-\infty }^\infty J_j X_i^{j} \Big)Z	}.
\end{align}
\end{definition}


There is a differential-algebraic expression that expresses $W^{(g)}_{m,n+1}$ in terms of $W^{(g)}_{m+1,n}$ and $W^{(g')}_{m',n'}$ with $g'\leq g$, $n'\leq n$, and $2g'+n'+m'\leq 2g+n+m$, see~\cite{BDKS-symplectic}. Let us recall it.

Let $M=\set{m}$ and $N=\set{m+n+1}\setminus \set{m+1}$. 
It is convenient to denote $X=X_{m+1}$ and $w=w_{m+1}$. Let
\begin{align} 
	\cT^{X}_{m+1,n}(u;X_M;X;w_N) &\coloneqq
	\sum_{k=1}^\infty\frac{\hbar^{2(k-1)}}{k!}\left(\prod_{i=1}^k
	\big\lfloor_{X_{\bar i}\to X}u\,\cS(u\hbar X_{\bar i}\partial_{X_{\bar i}})\right)
	\\ \notag & \quad
	\left( \sum_{g=0}^{\infty} \hbar^{2g} W^{(g)}_{m+k,n} 
	(X_{M},X_{\{\bar1,\dots,\bar k\}};w_{N})-
	\delta_{(m,k,n),(0,2,0)}\frac{X_{\bar1}X_{\bar2}}{(X_{\bar1}-X_{\bar2})^2}\right),
	\\
	\cW^{X}_{m+1,n}(u;X) & \coloneqq \frac{e^{\cT^X_{1,0}(u;X)}}{u\cS(u\,\hbar)}
	\sum_{\substack{\sqcup_\alpha K_\alpha=M\cup N,~
			K_\alpha\ne\varnothing,\\~I_\alpha=K_\alpha\cap M,~J_\alpha=K_\alpha\cap N}}
	\prod_{\alpha}\cT^{X}_{|I_\alpha|+1,|J_\alpha|}(u;X_{I_\alpha};X;w_{J_\alpha}).
\end{align}
In order to shorten the notation, we omit dependencies on $X_{\set m},w_{\set{m+n+1}\setminus\set{m+1}}$, and $\hbar$ for $\cW_{m+1,n}$ and on $\hbar$ for $\cT_{m+1,n}$, respectively, in the lists of their arguments. 

Let us also introduce the following notation:
\begin{align}	
	L_0(v,\theta)&:=e^{v\bigl(\frac{\cS(v\hbar\partial_\theta)}{\cS(\hbar\partial_\theta)}\hat\psi(\theta)-\psi(\theta)\bigr)},\\ \label{eq:Lr}
	L_r(v,\theta)&:=e^{-v\psi(\theta)}\partial_\theta^re^{v\psi(\theta)}L_0(v,\theta)=(\partial_\theta+v\psi'(\theta))^rL_0(v,\theta).
\end{align}

With the convention on omitting the arguments applied to $W^{(g)}_{m,n+1}$, we have the following proposition:

\begin{proposition} \label{prop:Wmnrec} We have:
	\begin{align}\label{eq:mainrecchange}
		W^{(g)}_{m,n+1}(w)& =[\hbar^{2g}]
		\sum_{j\ge0} (w\partial_w)^j[v^j]\Bigl(\sum_{r\ge0}L_r(v,W^{(0)}_{1,0}(X))[u^r]\frac{d\log X}{d \log w}e^{-u\,W^{(0)}_{1,0}(X)}\cW^{X}_{m+1,n}(u;X)
		\\ \notag & \qquad \qquad \qquad \qquad
		+\delta_{m+n,0}\frac{L_0(v,W^{(0)}_{1,0}(X))}{v}\frac{dW^{(0)}_{1,0}(X)}{d\log w}
		+\delta_{m+n,0} W^{(0)}_{1,0}(X) \Bigr),
	\end{align}	
	where $X$ and $w$ are related by
	\begin{equation}\label{eq:xychange}	
	w = X e^{-\psi(W^{(0)}_{1,0}(X))}.
	\end{equation}
\end{proposition}
\begin{corollary}
	For $(g,m+n)\neq(0,1)$, $W^{(g)}_{m,n}$ can be extended to a global meromorphic function on $\Sigma^{m+n}$, and Equation~\eqref{eq:mainrecchange} then holds for these global meromorphic functions.
\end{corollary}
Similarly to how we do it with $\omega^{(g)}_{m,0}$'s and $W^{(g)}_{m,0}$'s, we abuse notation for $W^{(g)}_{m,n}$'s as well and use one and the same notation for their expansions and for their global versions.

Moreover, if we introduce the notation
\begin{align} \label{eq:cwdef1}
	\cT^{w}_{m,n+1}(u;X_M;w;w_N) &\coloneqq
	\sum_{k=1}^\infty\frac{\hbar^{2(k-1)}}{k!}\left(\prod_{i=1}^k
	\big\lfloor_{w_{\bar i}\to w}u\,\cS(u\hbar w_{\bar i}\partial_{w_{\bar i}})\right)
	\\ \notag & \quad
	\left(
	\sum_{g=0}^{\infty} \hbar^{2g} W^{(g)}_{m,n+k}(X_{M};w_{\{\bar1,\dots,\bar k\}},w_{N})-
	\delta_{(m,k,n),(0,2,0)}\frac{w_{\bar1}w_{\bar2}}{(w_{\bar1}-w_{\bar2})^2}\right),
	\\ \label{eq:cwdef2}
	\cW^{w}_{m,n+1}(u;w) & \coloneqq \frac{e^{\cT^w_{0,1}(u;w)}}{u\cS(u\,\hbar)}
	\sum_{\substack{\sqcup_\alpha K_\alpha=M\cup N,~
			K_\alpha\ne\varnothing,\\~I_\alpha=K_\alpha\cap M,~J_\alpha=K_\alpha\cap N}}
	\prod_{\alpha}\cT^{w}_{|I_\alpha|,|J_\alpha|+1}(u;X_{I_\alpha};w;w_{J_\alpha}),
\end{align}
and also
\begin{equation}
	\tilde L_r(v,\theta,u):=(\partial_\theta+v\psi'(\theta))^re^{u\,\theta}L_0(v,\theta),
\end{equation}
then we have
\begin{proposition}	\label{prop:cWmnrec}
	\begin{align}\label{eq:WYtoomega}
		W^{(g)}_{m+1,n}(X)& =[\hbar^{2g}]
		\sum_{j\ge0} (X\partial_X)^j[v^j]\Bigl(\sum_{r\ge0}L_r(-v,W^{(0)}_{0,1}(w))[u^r]\frac{d\log w}{d \log X}e^{-u\,W^{(0)}_{0,1}(w)}\cW^{w}_{m,n+1}(u;w)
		\\ \notag & \qquad \qquad \qquad \qquad
		-\delta_{m+n,0}\frac{L_0(-v,W^{(0)}_{0,1}(w))}{v}\frac{dW^{(0)}_{0,1}(w)}{d\log X}
		+\delta_{m+n,0} W^{(0)}_{0,1}(w) \Bigr),
	\end{align}
	\begin{align}\label{eq:WwtoW}
		\cW^{w}_{m,n+1}(\tilde u;w)& =
		\sum_{j\ge0} (w\partial_w)^j [v^j]\Bigl(\sum_{r\ge0}\tilde L_r(v,W^{(0)}_{1,0}(X),\tilde u)[u^r]\frac{d \log X}{d\log w}e^{-u\,W^{(0)}_{1,0}(X)}\cW^{X}_{m+1,n}(u;X)
		\\ \notag & \qquad \qquad \qquad
		+\delta_{m+n,0}\frac{\tilde L_0(v,W^{(0)}_{1,0}(X),\tilde u)}{v}\frac{dW^{(0)}_{1,0}(X)}{d \log w}+\delta_{m+n,0}\frac{e^{\tilde u W^{(0)}_{1,0}(X)}}{\tilde u}\Bigr),
		\\ \label{eq:WXtoW}
		\cW^{X}_{m+1,n}(u;X)& =
		\sum_{j\ge0} (X\partial_X)^j[v^j]\Bigl(\sum_{r\ge0}\tilde L_r(-v,W^{(0)}_{0,1}(w),u)[\tilde u^r]\frac{d\log w}{d\log X}e^{-\tilde u\,W^{(0)}_{0,1}(w)}\cW^{w}_{m,n+1}(\tilde u;w)
		\\ \notag & \qquad \qquad \qquad
		-\delta_{m+n,0}\frac{\tilde L_0(-v,W^{(0)}_{0,1}(w),u)}{v}\frac{dW^{(0)}_{0,1}(w)}{d\log X}+\delta_{m+n,0}\frac{e^{u W^{(0)}_{0,1}(w)}}{u}\Bigr).
	\end{align}
\end{proposition}

The proofs of Propositions~\ref{prop:Wmnrec} and~\ref{prop:cWmnrec} go completely analogously to \cite[Section~4.4]{ABDKS-XYSwap}. Same as with Equation~\eqref{eq:mainrecchange}, Equations~\eqref{eq:WYtoomega}--\eqref{eq:WXtoW} hold for both the expansions of the functions $W^{(g)}_{m,n}$ and $\cW^{(g)}_{m,n}$ as well as for their global versions.

Let us also introduce the following notation:
\begin{notation}
	\begin{align}
		\cW^{X,(g)}_{m,n}&:=[\hbar^{2g}]		\cW^{X}_{m,n}, \\
		\cW^{w,(g)}_{m,n}&:=[\hbar^{2g}]		\cW^{w}_{m,n}.
	\end{align}
\end{notation}

We note an obvious property of $\cW^{X,(g)}_{m,n}$ and $\cW^{w,(g)}_{m,n}$:
\begin{proposition}
	\begin{align}
		[u^0]\cW^{X,(g)}_{m,n}(u;X_{\set{m}};w_{\set{m+n}\setminus\set{n}})
		& =[u^0]\cW^{w,(g)}_{m,n}(u;X_{\set{m}};w_{\set{m+n}\setminus\set{n}})
		\\ \notag &
		=W^{(g)}_{m,n}(X_{\set{m}};w_{\set{m+n}\setminus\set{n}}).
	\end{align}
\end{proposition}

\section{Topological recursion}\label{sec:toprec}
In this section we prove our main Theorem \ref{thm:mainthm}. 

\subsection{Loop equations}
Let us at first reformulate linear and quadratic loop equations described in Definition~\ref{def:trloopeqdef} in a different form (see the discussion on the equivalence of the two definitions of loop equations in \cite[Section 5]{ABDKS-XYSwap}).

\begin{notation}\label{def:xispace} We denote by $\Xi^X$ the space of functions defined in a neighborhood of the zero locus of $dX$ on $\Sigma$ and having odd polar part with respect to the involution~$\sigma$:
	\begin{equation} \label{eq:Xi}
		f\in\Xi^X\quad\Leftrightarrow\quad \forall i\quad f(z)+f(\sigma_i(z))\text{ is holomorphic at $z\to p_i$.}
	\end{equation}
\end{notation}
Any function with poles of order at most one belongs to $\Xi^X$. Besides, $\Xi^X$ is preserved by the operator $X\partial_X$. It follows that any function of the form $(X\partial_X)^kf$, $k\ge0$, belongs to $\Xi^X$ if $f$ has at most simple poles at the points $p_i$, $i\in\{1,\dots,N\}$. 

\begin{definition}\label{def:mnlqloop}
We say that the $(m,n)$-point functions $\{W^{(g)}_{m,n}\}$ satisfy the linear and quadratic loop equations, if
\begin{align}\label{eq:Xloop}
	[u^0]\cW^{X,(g)}_{m+1,n}(u;X),\;[u^1]\cW^{X,(g)}_{m+1,n}(u;X) \in \Xi^X
\end{align}
 and 
\begin{align}\label{eq:wloop}
	[\tilde u^0]\cW^{w,(g)}_{m,n+1}(\tilde u;w),\; [\tilde u^1]\cW^{w,(g)}_{m,n+1}(\tilde u;w) \in \Xi^w.
\end{align} 
\end{definition}

Note that this set of equations contains (for $n=0$ and $m=0$ respectively) the usual linear and quadratic loop equations for $\{W^{(g)}_{m,0}\}$ and $\{W^{(g)}_{0,n}\}$ on the curves 
$\left(\Sigma,X,y,B\right)$ and $\left(\Sigma,w,y,B\right)$ respectively, in the sense of~\eqref{eq:lloop}--\eqref{eq:qloop}.

\begin{proposition}\label{prop:looppoles} The $(m,n)$-point functions $W^{(g)}_{m,n}$ satisfy linear and quadratic loop equations in the sense of Definition~\ref{def:mnlqloop}. Moreover, $W^{(g)}_{m,n}$ is holomorphic in $z_i$ at any zero of $dw$ for $i\in\set{m}$, and it is holomorphic in $z_j$ at any zero of $dX$ if $j\in \set{m+n}\setminus\set{m}$.
\end{proposition}

\begin{proof} This proposition mostly follows~\cite[Proposition~5.7]{ABDKS-XYSwap}, but there are some differences as we use another form of recursive formulas, so we repeat the reasoning here with respective adjustments.
	
 The fact that $W^{(g)}_{m,n}$ is holomorphic in $z_i$ at any zero of $dw$ for $i\in\set{m}$ is obvious: these poles are absent for the initial functions~$W^{(g)}_{m,0}$ and they do not appear in the course of recursion~\eqref{eq:mainrecchange}. 
	Then, the very form of Equation~\eqref{eq:WwtoW} implies that the loop equations 
	\eqref{eq:wloop}
	hold for 
	all $(g,m,n)$.

	Let us look more attentively at the identity~\eqref{eq:WXtoW} for the case $n=0$. The summand with $r=0$ corresponds to the coefficient $[\tilde u^0]\cW^{w,(g)}_{m,1}(\tilde u;w)=W^{(g)}_{m,1}$. Also note that $[u^0]\tilde{L}_r(v,\theta,u)=L_r(v,\theta)$, $[v^0]L_r(v,\theta)=\delta_{r,0}$, and $[v^0u^1]\tilde{L}_r(v,\theta,u)=\delta_{r,0}\,\theta +\delta_{r,1}$. So we can represent the coefficients of $u^0$ and $u^1$ of~\eqref{eq:WXtoW} for $n=0$ as follows (for $(g,m)\neq(0,0)$): 
	\begin{align} \label{eq:34}
		[u^0]\cW^{X,(g)}_{m+1,0}(u;X)&=\dfrac{d \log w}{d\log X}W^{(g)}_{m,1}+\mathrm{Rest_0}
		,\\ \label{eq:35}
		[u^1]\cW^{X,(g)}_{m+1,0}(u;X)&=\dfrac{d \log w}{d\log X}W^{(0)}_{1,0}W^{(g)}_{m,1}+\mathrm{Rest_1},
	\end{align}
where
\begin{align} \label{eq:rest0}
	\mathrm{Rest_0}&=
	[\hbar^{2g}]\sum_{j\ge1} (X\partial_X)^j[v^j]\Bigg(\sum_{r\ge0} L_r(-v,W^{(0)}_{0,1}(w))[\tilde u^r]\frac{d\log w}{d\log X}e^{-\tilde u\,W^{(0)}_{0,1}(w)}\cW^{w}_{m,1}(\tilde u;w)
	\\ \notag &\phantom{=[\hbar^{2g}]\sum_{j\ge1} (X\partial_X)^j[v^j]\Bigg(}
	-\delta_{m,0}\frac{L_0(-v,W^{(0)}_{0,1}(w))}{v}\frac{dW^{(0)}_{0,1}(w)}{d\log X}\Bigg)\\  \notag 
	&\phantom{=}-\delta_{m,0}[\hbar^{2g}v^1]L_0(-v,W^{(0)}_{0,1}(w))\frac{dW^{(0)}_{0,1}(w)}{d\log X}
	,\\ \label{eq:rest1}
	\mathrm{Rest_1}&=[\hbar^{2g}]\sum_{j\ge1} (X\partial_X)^j[v^j]\Bigg(\sum_{r\ge0}[\tilde u^1]\tilde L_r(-v,W^{(0)}_{0,1}(w),u)[\tilde u^r]\frac{d\log w}{d\log X}e^{-\tilde u\,W^{(0)}_{0,1}(w)}\cW^{w}_{m,1}(\tilde u;w) 
	\\ \notag & \phantom{=[\hbar^{2g}]\sum_{j\ge1} (X\partial_X)^j[v^j]\Bigg(}
	-\delta_{m,0}\frac{[\tilde u^1]\tilde L_0(-v,W^{(0)}_{0,1}(w),u)}{v}\frac{dW^{(0)}_{0,1}(w)}{d\log X}\Bigg)\\ \notag
	&\phantom{=}+[\hbar^{2g}\tilde u^1]\frac{d\log w}{d\log X}e^{-\tilde u\,W^{(0)}_{0,1}(w)}\cW^{w}_{m,1}(\tilde u;w)\\ \notag
	&\phantom{=}-\delta_{m,0}[\hbar^{2g}v^1u^1]\tilde L_0(-v,W^{(0)}_{0,1}(w),u)\frac{dW^{(0)}_{0,1}(w)}{d\log X}.
\end{align}

	By the linear and quadratic loop equations~\eqref{eq:lloop}--\eqref{eq:qloop} for $W^{(g)}_{m,0}$ 
	the left hand sides of Equations~\eqref{eq:34}--\eqref{eq:35} belong to~$\Xi^X$, cf.~\eqref{eq:Xi}.
	Let us study~\eqref{eq:rest0}--\eqref{eq:rest1}. Note that the right hand sides of these formulas
	 involve only functions $W^{(g')}_{m',n'}$ with $2g'-2+m'+n'\leq2g-2+m$. Indeed, from~\eqref{eq:cwdef1}--\eqref{eq:cwdef2} note that $e^{-\tilde u\,W^{(0)}_{0,1}(w)}$ gets canceled with a factor coming from $e^{\cT^w_{0,1}(u;w)}$ in $\cW^{w}_{m,1}(\tilde u;w)$, and $\forall k\in\mathbb{Z}_{\geq 0}\;\; [\hbar^{k}]e^{-\tilde u\,W^{(0)}_{0,1}(w)}\cW^{w}_{m,1}(\tilde u;w)$ is a finite polynomial in $\tilde u$. The way $\hbar$ enters the expressions in~\eqref{eq:rest0}--\eqref{eq:rest1} guarantees that after taking $[\hbar^{2g}]$ we are left only with functions $W^{(g')}_{m',n'}$ with $2g'-2+m'+n'\leq2g-2+m$. In particular, note that  for $j\geq 1$ the expression $[v^j]\tilde L_0(-v,W^{(0)}_{0,1}(w),u)$ is proportional at least to $\hbar^1$, and thus we do not get any instances of e.g. $W^{(g)}_{m,1}$ in the first lines of~\eqref{eq:rest0} and~\eqref{eq:rest1}. We also do not get any instance of e.g. $W^{(g)}_{m,1}$ from the second-to-last line of~\eqref{eq:rest1} since we take $[\tilde u^1]$ there, and after the cancellation of $e^{-\tilde u\,W^{(0)}_{0,1}(w)}$ (as mentioned above) all of the remaining respective terms in $\cW^{w}_{m,1}(\tilde u;w)$ contain extra powers of $\hbar$ etc.
	 
	  So, arguing by induction, we may assume that the regularity of all functions $W^{(g')}_{m',n'}$ entering~\eqref{eq:rest0}--\eqref{eq:rest1} at zeros of~$dX$ with respect to the variables $z_j$, $j\in \set{m'+n'}\setminus\set{m'}$, is already proved. It follows that $\mathrm{Rest_0}$ and $\mathrm{Rest_1}$ also belong to $\Xi^X$, since $X\partial_X$ preserves $\Xi^X$ and all the other factors entering these expressions are regular at the  
	zero points of $dX$.
	We conclude that $W^{(g)}_{m,1}$ regarded as a function of the last argument $z_{m+1}=z$ satisfies
	\begin{equation} \label{eq:LLE-QLE-holo}
		\frac{d \log w}{d\log X}W^{(g)}_{m,1}\in\Xi^X,\qquad
		W^{(0)}_{1,0}\frac{d\log w}{d\log X}W^{(g)}_{m,1}\in\Xi^X.
	\end{equation}
	Let us show that this implies that $W^{(g)}_{m,1}$ is holomorphic at a given 
	zero point~$p$ of $dX$. Indeed, the first condition implies that if $\frac{d\log w}{d\log X}W^{(g)}_{m,1}$ has a pole at~$p$ then this pole should be of odd order. On the other hand, multiplying it by $W^{(0)}_{1,0}-W^{(0)}_{1,0}(p)$ we again obtain an element of~$\Xi^X$, but the order of pole decreases by~$1$ changing the parity of its order (since $W^{(0)}_{1,0}$ is regular at $p$). This is only possible if $\frac{d \log w}{d\log X}W^{(g)}_{m,1}$ has at most simple pole, that is, $W^{(g)}_{m,1}$ is regular.
	
	Thus, we have established regularity of  $W^{(g)}_{m,1}$ in $z_{m+1}$ at any zero of~$dX$. Applying the recursion~\eqref{eq:mainrecchange} we obtain that $W^{(g)}_{m+1-n,n}$ is regular in $z_{m+1}$  at zeroes of $dX$ for any $n=1,2,\dots,m+1$. On the other hand, since $W^{(g)}_{m,n}$ is symmetric with respect to $z_{m+1},\dots,z_{m+n}$ (as can be seen by iteratively applying the recursive formula~\eqref{eq:mainrecchange} expressing everything in terms of $W^{g}_{m,0}$, obtaining a formula similar to~\eqref{eq:SymplDualExplicit}, which is symmetric), we established thereby the regularity of $W^{(g)}_{m,n}$ with respect to all last~$n$ arguments at zeroes of~$dX$ by induction in $2g-2+m+n$.
	
	Same as with the poles at the zeroes of $dw$ which we mentioned at the beginning of this proof, this implies the loop equations~\ref{eq:Xloop} 
	for all $(g,m,n)$ 
	by the very form of~\eqref{eq:WwtoW}.
\end{proof}
\subsection{Projection property}
Now we need to prove the remaining part of the projection property that can be formulated for the functions $W^{(g)}_{0,n}$:
\begin{proposition}\label{prop:projpr}
	The functions $W^{(g)}_{0,n}$ defined by~\eqref{eq:WsdZ} satisfy the projection property. That is, the respective differentials $\omega^{(g)}_{0,n+1}(z,z_{\llbracket n\rrbracket})$, for $2g-1+n>0$, have no poles in $z$ other than $p_1,\dots,p_N$.
\end{proposition}
\begin{proof}
	The proof goes along the lines of \cite[Section 5.3]{ABDKS-XYSwap}. One difference is that there may be poles of $\omega^{(g)}_{0,n+1}$ in $z$ at the poles of $\psi(y(z))$ coming from the zeroes and poles of the rational function $R(\theta)$. But these poles are canceled by the arguments completely analogous to the ones of~\cite[Lemma 4.4]{BDKS-toporec-KP}.
	
	The possible poles of $\omega^{(g)}_{0,n+1}$ in $z$ at the poles of $y(z)$ get canceled analogously to~\cite[Lemma 4.3]{BDKS-toporec-KP} in the case if $X$ does not a have a pole at this given pole of $y$, and analogously to~\cite[Lemma 4.6]{BDKS-toporec-KP} if $dX/X$ has a simple pole there.

	Note that 
	\begin{itemize}
		\item in these analogies the pair $(w,X)$ from the present paper works as the pair $(X,z)$ from~\cite{BDKS-toporec-KP};

		\item a part of the \emph{enumerative-type spectral curve condition} that concerns poles of $y$ was specifically added to the statement to make the references to~\cite[Lemmas 4.3 and 4.6]{BDKS-toporec-KP} possible and to avoid further analysis at this point;

		\item for these cancellations to happen we need exactly the relation between $\hat\psi$ and $\psi$ that we have (Equations~\eqref{eq:hatpsi} and~\eqref{eq:psi}).
	\end{itemize}
	\end{proof}

\subsection{Proof of Theorem~\ref{thm:mainthm}}
	\begin{proof}[Proof of Theorem~\ref{thm:mainthm}]
	Consider the $n$-point functions $W^{(g)}_{0,n}$ defined by their expansions~\eqref{eq:WsdZ}.
	
	Proposition~\ref{prop:looppoles} implies that the $n$-point functions $W^{(g)}_{0,n}$ satisfy the linear and quadratic loop equations of Definition~\ref{def:trloopeqdef} for the spectral curve $(\Sigma,w,y,B)$, where $w(z)$ is given by~\eqref{eq:wofz}. Proposition~\ref{prop:projpr} states that they satisfy the projection property. This means that they are precisely the $n$-point functions produced by the topological recursion on this spectral curve. 
	
	 As mentioned in Section~\ref{sec:expansion}, \cite[Theorem~4.14]{BDKS-FullySimple} implies that $W^{(g)}_{0,n}$ of~\eqref{eq:WsdZ} can be expressed in terms of the functions $W^{(g)}_{m,0}$ of~\eqref{eq:WZ} exactly via the formula~\eqref{eq:SymplDualExplicit}. This completes the proof.
 \end{proof}


\printbibliography
\end{document}